\documentclass[english]{ceurart}
\usepackage[utf8]{inputenc}
\usepackage[T1]{fontenc}
\usepackage{babel}
\usepackage{stmaryrd}

\usepackage{bussproofs}
\EnableBpAbbreviations

\usepackage{graphbox}
\DeclareRobustCommand{\DLLogo}{%
  \begingroup\normalfont
  \kern-1.75pt\includegraphics[align=c,height=1.25\baselineskip]{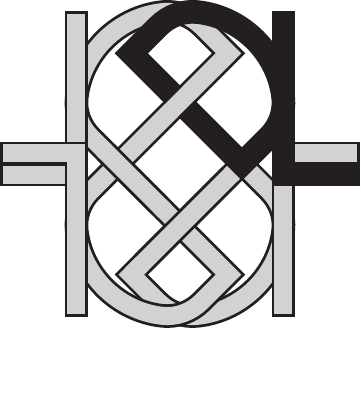}\kern-1.5pt%
  \endgroup
}

\usepackage{amsthm}
\newtheorem{theorem}{Theorem}
\newtheorem{definition}{Definition}

\usepackage{tikz}
\usetikzlibrary{shapes,trees}

\usepackage{amssymb}
\usepackage{amsmath}
\usepackage{xspace}
\usepackage{todonotes}
\usepackage{hyperref}
\usepackage{mathtools}
\usepackage{graphicx}
\usepackage{pgfplots}
\usepackage{csquotes}
\usepackage{framed}
\usepackage{subcaption}
\usepackage{thm-restate}

\newcommand{\ttodo}[4]{\ifthenelse{\equal{#1}{inline}}{\todo[author=#2,color=#3]{#4}}{\todo[color=#3]{#2: #4}}}

\newcommand{\sideConColor}[1]{\textcolor{black!40!white}{#1}}

\def\define#1#2#3%
{%
\renewcommand*{\do}[1]{%
 \expandafter\newcommand\csname
 #1\endcsname{#2}
}
\docsvlist{#3}
}

\define{#1}
{{\textsc{#1}}\xspace}
{Elk,Lethe,Fame,Evonne,Evee,HermiT,Spass,Capi,GraphStream,Nemo,Konclude,Sequoia,Envelope,Textbook}

\newcommand{\ie}{i.e.\ }
\newcommand{\eg}{e.g.\ }

\newcommand{\NC}{\ensuremath{\textsf{N}_\textsf{C}}\xspace}
\newcommand{\NR}{\ensuremath{\textsf{N}_\textsf{R}}\xspace}

\define{#1mc}
{{\ensuremath{\mathcal{#1}}}\xspace}
{A,B,C,D,E,F,G,H,I,J,K,L,M,N,O,P,Q,R,S,T,U,V,W,X,Y,Z}

\define{#1}
{{\textsc{#1}}\xspace}
{NP,coNP,PSpace,ExpTime,ExpSpace}

\define{#1}
{{\ensuremath{\mathcal{#1}}}\xspace}
{ALC,EL,ALCOI,ALCH,ELH,ELI,SROIQ}

\newcommand{\ELHb}{\ensuremath{\mathcal{ELH}_\bot}\xspace}

\newcommand{\ELpb}{\ensuremath{\mathcal{EL}^+_\bot}\xspace}
\newcommand{\ELpp}{\ensuremath{\mathcal{EL}^{++}}\xspace}

\newcommand{\Protege}{Prot{\'e}g{\'e}\xspace}

\newcommand{\opfont}[1]{\text{\sf{#1}}} %
\newcommand{\ltuple}{\langle}
\newcommand{\rtuple}{\rangle}
\newcommand{\tuple}[1]{\ltuple{#1}\rtuple}

\newcommand{\cut}[1]{\opfont{cut}(#1)}
\newcommand{\dcw}[1]{\opfont{cw}(#1)}

\newcommand{\blue}[1]{\color{blue}#1\color{black}\xspace}

\newcommand{\mypar}[1]{\smallskip\noindent\textbf{#1.}\ }

\renewcommand{\blue}[1]{#1\xspace}
\begin{document}

\copyrightyear{2025}
\copyrightclause{Copyright for this paper by its authors.
  Use permitted under Creative Commons License Attribution 4.0
  International (CC BY 4.0).}

\conference{\DLLogo{} DL 2025: 38th International Workshop on Description Logics, September 3--6, 2025, Opole, Poland}

\title{The Shape of EL Proofs: A Tale of Three Calculi (Extended Version)}[The Shape of \EL Proofs: A Tale of Three Calculi (Extended Version)]

\author[]{Christian Alrabbaa}[%
orcid=0000-0002-2925-1765,
email=christian.alrabbaa@tu-dresden.de,
]
\author[]{Stefan Borgwardt}[%
orcid=0000-0003-0924-8478,
email=stefan.borgwardt@tu-dresden.de,
]
\author[]{Philipp Herrmann}[%
email=philipp.herrmann1@tu-dresden.de,
]
\author[]{Markus Kr{\"o}tzsch}[%
orcid=0000-0002-9172-2601,
email=markus.kroetzsch@tu-dresden.de
]
\address[]{Institute of Theoretical Computer Science, Technische Universität Dresden, 01062 Dresden, Germany}

\begin{abstract}
	Consequence-based reasoning can be used to construct proofs that explain entailments of description logic (DL) ontologies.
	In the literature, one can find multiple consequence-based calculi for reasoning in the \EL family of DLs, each of which gives rise to proofs of different shapes.
	Here, we study three such calculi and the proofs they produce on a benchmark based on the OWL Reasoner Evaluation.
	The calculi are implemented using a translation into existential rules with stratified negation, which had already been demonstrated to be effective for the calculus of the \Elk reasoner.
	We then use the rule engine \Nemo to evaluate the rules and obtain traces of the rule execution.
	After translating these traces back into DL proofs, we compare them on several metrics that reflect different aspects of their complexity.
\end{abstract}


\maketitle

\section{Introduction}
\label{sec:introduction}

Consequence-based reasoning for DLs has a long tradition, starting from the first reasoning algorithms for $\mathcal{EL}$ with general concept 
inclusions~\cite{DBLP:conf/ecai/Brandt04,DBLP:conf/ijcai/BaaderBL05}, all the way up to expressive DLs like 
$\mathcal{ALCHOIQ}$~\cite{DBLP:conf/ijcai/CucalaGH18,DBLP:conf/birthday/CucalaGH19}.
Due to its favorable computational properties, it is employed in several reasoners, such as \Elk, \Konclude, and \Sequoia~\cite{DBLP:journals/jar/KazakovKS14,DBLP:journals/ws/SteigmillerLG14,DBLP:journals/jair/BateMGCSH18}.
Another advantage of consequence-based approaches is the possibility to extract proofs consisting of step-wise derivations, in order to explain consequences of the ontology to an ontology engineer.
In principle, the relatively simple rules of a reasoning calculus can immediately be used to construct proofs.
However, in practice, this would have to be implemented by each consequence-based reasoner, and is so far only supported by \Elk \cite{DBLP:conf/semweb/KazakovK14}.
Several approaches have been developed to extract proofs from reasoners in a black-box fashion~\cite{DBLP:conf/semweb/HorridgePS10,DBLP:conf/jelia/Schlobach04,DBLP:conf/lpar/AlrabbaaBBKK20}, which do not use a fixed set of rules and often lead to more complicated proof steps, but smaller proofs overall~\cite{DBLP:conf/kr/AlrabbaaBFHKKKK24}.
Our main research question is how proofs resulting from consequence-based calculi differ based on the shape of the rules.
\blue{Specifically, we explore which calculi may be more appropriate in certain scenarios---for instance, producing proofs suitable for specific types of visualizations,
or ones in which information is introduced and resolved locally, resulting in a reasoning structure that is easier to follow.}
We investigate this on three calculi for the \EL family of description logics~\cite{DBLP:conf/ijcai/BaaderBL05,DBLP:journals/jar/KazakovKS14,BHLS-17}.

We follow an approach to encode DL reasoning into (an extension of) Datalog rules~\cite{DBLP:conf/ijcai/Krotzsch11,DBLP:journals/ki/CarralDK20,DBLP:conf/kr/Ivliev0MSK24}.
To execute these rules, we use \Nemo, a Datalog-based rule engine that is easy to use and offers many advanced features, such as datatypes, existential rules, aggregates, and stratified negation~\cite{DBLP:conf/kr/Ivliev0MSK24}.
This has already been demonstrated to be effective for the reasoning calculus of \Elk, including pre-processing of the ontology in OWL format,\footnote{https://www.w3.org/TR/owl2-overview/} which essentially implements an \ELpb reasoner by existential rules with stratified negation~\cite{DBLP:conf/kr/Ivliev0MSK24}.
The advantage of this approach is that one can relatively quickly implement new reasoning procedures with low effort.

In this paper, we also consider two other calculi for \EL-based logics from the literature: the original \ELpp calculus we denote by \Envelope~\cite{DBLP:conf/ijcai/BaaderBL05} (restricted to \ELpb, \ie without nominals and concrete domains) and the \EL calculus \Textbook~\cite{BHLS-17} (extended to \ELHb, \ie with role hierarchies and $\bot$).
The modifications to the calculi allow us to compare them on a dataset of $\ELHb$ reasoning tasks extracted from the 2015 OWL Reasoner Evaluation (ORE)~\cite{DBLP:journals/jar/ParsiaMGGS17}.
After encoding them into rules as for the \Elk calculus, we use the tracing capabilities of \Nemo to compute proofs for the derived consequences.
However, since the \Nemo traces are based on the translated rules and the calculi may use statements that are not expressible as DL axioms, \eg in side conditions, we still need to transform these traces before we obtain proofs over DL statements that can be inspected by ontology engineers.
Finally, we compare the shape of the resulting proofs using several measures, such as the
\emph{size}, \emph{depth}, \emph{directed cutwidth}~\cite{DBLP:journals/jcss/BodlaenderFT09}, and cognitive 
complexity of inference steps~\cite{DBLP:journals/kbs/HorridgeBPS13}.

\section{Calculi and Proofs}
\label{sec:preliminaries}

Let \NC and \NR be two disjoint, countably infinite sets of \emph{concept-} and \emph{role names}, respectively.
\emph{\ELpb concepts} are defined by the grammar $C,D::=\top\mid\bot\mid A\mid C\sqcap D\mid \exists r.C$, where $A\in\NC$ and $r\in\NR$.
\emph{\ELpb axioms} are either \emph{concept inclusions} of the form $C\sqsubseteq D$ or \emph{complex role inclusions} of the form $r_1\circ \dots\circ r_n\sqsubseteq r$, where $r_1,\dots,r_n,r\in\NR$.
An \emph{\ELpb TBox} is a finite set of \ELpb axioms.
\ELHb is the fragment of \ELpb that only allows \emph{simple role inclusions} of the form $r\sqsubseteq s$.
We assume the reader to be familiar with the semantics of these logics, in particular the definition of \emph{entailment} of an axiom~$\alpha$ from a TBox~\Tmc, written $\Tmc\models\alpha$~\cite{BHLS-17}.

\mypar{Calculi}
We consider three inference calculi for fragments of \ELpb that are tailored towards \emph{classification}, \ie computing all entailments of the form $\Tmc\models A\sqsubseteq B$ for $A,B\in\NC$.

\OuterFrameSep0pt

\begin{figure}[tb]
	\vskip-\smallskipamount
	\begin{framed}
  	\centering
		\AXC{$\mathsf{init}(C)$}
		\LeftLabel{$\mathsf{R}_0$}
		\UIC{$C\sqsubseteq C$}
		\DP
		\quad
		\AXC{$\mathsf{init}(C)$}
		\LeftLabel{$\mathsf{R}_\top$}
		\RightLabel{\sideConColor{: $\top$ occurs negatively in \Tmc}}
		\UIC{$C\sqsubseteq\top$}
		\DP
		\quad
		\AXC{$C\sqsubseteq D$}
		\LeftLabel{$\mathsf{R}_\sqsubseteq$}
		\RightLabel{\sideConColor{: $D \sqsubseteq E \in \Tmc$}}
		\UIC{$C\sqsubseteq E$}
		\DP
		\\[3ex]
		\AXC{$C\sqsubseteq D_1 \sqcap D_2$}
		\LeftLabel{$\mathsf{R}_{\sqcap}^-$}
		\UIC{$C\sqsubseteq D_1 \quad C\sqsubseteq D_2$}
		\DP
		\quad
		\AXC{$C\sqsubseteq D_1$} 
		\AXC{$C\sqsubseteq D_2$}
		\LeftLabel{$\mathsf{R}_\sqcap^+$}
		\RightLabel{\sideConColor{: $D_1\sqcap D_2$ occurs negatively in \Tmc}}
		\BIC{$C\sqsubseteq D_1 \sqcap D_2$}
		\DP
		\\[3ex]
		\AXC{$C\sqsubseteq\exists r.E$}
		\LeftLabel{$\mathsf{R}_\exists^-$}
		\UIC{$C\xrightarrow{r} E$}
		\DP
		\quad
		\AXC{$C\xrightarrow{r} D$}
		\AXC{$D \sqsubseteq E$}
		\LeftLabel{$\mathsf{R}_\exists^+$}
		\RightLabel{\sideConColor{: \parbox{4cm}{$r \sqsubseteq^*_\Tmc s$ \\ $\exists s.E$ occurs negatively in \Tmc}}}
		\BIC{$C \sqsubseteq \exists s.E$}
		\DP
		\\[3ex]
		\AXC{$C\xrightarrow{r} E$}
		\LeftLabel{$\mathsf{R}_\rightsquigarrow$}
		\UIC{$\mathsf{init}(E)$}
		\DP
		\quad
		\AXC{$C\xrightarrow{r} E$}
		\AXC{$E \sqsubseteq \bot$}
		\LeftLabel{$\mathsf{R}_\bot$}
		\BIC{$C \sqsubseteq \bot$}
		\DP
		\quad
		\AXC{$C\xrightarrow{r_1} D$}
		\AXC{$D\xrightarrow{r_2} E$}
		\LeftLabel{$\mathsf{R}_\circ$}
		\RightLabel{\sideConColor{: \parbox{2cm}{$r_1 \sqsubseteq^*_\Tmc s_1$ \\ $r_2 \sqsubseteq^*_\Tmc s_2$ \\ $s_1 \circ s_2 \sqsubseteq s \in \Tmc$}}}
		\BIC{$C\xrightarrow{s} E$}
		\DP
	\end{framed}
	\vskip-\smallskipamount
	\caption{Optimized \Elk calculus~\cite{DBLP:journals/jar/KazakovKS14}.}
	\label{fig:elk-calculus}
\end{figure}
Figure~\ref{fig:elk-calculus} shows the inference rules used in the reasoner \Elk for \ELpb TBoxes.
Side conditions are marked in gray, where \enquote{$C$ occurs negatively} means that $C$ occurs within a concept on the left-hand side of some axiom in~\Tmc, and $\sqsubseteq_\Tmc^*$ denotes the precomputed \emph{role hierarchy}, \ie the transitive closure over simple role inclusions.
Furthermore, complex role inclusions are assumed to be in the normal form $r_1\circ r_2\sqsubseteq r$, using only binary role composition.
Axioms with longer role compositions can be normalized with the help of fresh role names.
Then, $\Tmc\models A\sqsubseteq B$ holds iff either $A\sqsubseteq B$ or $A\sqsubseteq\bot$ can be derived by these rules from $\mathsf{init}(A)$~\cite{DBLP:journals/jar/KazakovKS14}.

\begin{figure}[tb]
	\vskip-\smallskipamount
	\begin{framed}
  		\centering
		\AXC{\phantom{A}}
		\LeftLabel{$\mathsf{CR1}$}
		\UIC{$C\sqsubseteq C$}
		\DP
		\quad
		\AXC{\phantom{A}}
		\LeftLabel{$\mathsf{CR2}$}
		\UIC{$C\sqsubseteq\top$}
		\DP
		\quad
		\AXC{$C \sqsubseteq D_1$}
		\AXC{$C \sqsubseteq D_2$}
		\AXC{$D_1 \sqcap D_2 \sqsubseteq E$}
		\LeftLabel{$\mathsf{CR4}$}
		\TIC{$C\sqsubseteq E$}
		\DP
		\\[3ex]
		\AXC{$C \sqsubseteq D$}
		\AXC{$D \sqsubseteq E$}
		\LeftLabel{$\mathsf{CR3}$}
		\BIC{$C\sqsubseteq E$}
		\DP
		\quad
		\AXC{$C \sqsubseteq \exists r.D_1$\quad$D_1 \sqsubseteq D_2$\quad$\exists s.D_2 \sqsubseteq E$}
		\LeftLabel{$\mathsf{CR5'}$}
		\RightLabel{\sideConColor{: $r\sqsubseteq_\Tmc^* s$}}
		\UIC{$C\sqsubseteq E$}
		\DP
		\\[3ex]
		\AXC{$C \sqsubseteq \exists r.E$}
		\AXC{$E \sqsubseteq \bot$}
		\LeftLabel{$\mathsf{R}_\bot'$}
		\BIC{$C \sqsubseteq \bot$}
		\DP
	\end{framed}
	\vskip-\smallskipamount
	\caption{\Textbook calculus~\cite{BHLS-17} with a modified $\mathsf{CR5}$ and added variant of $\mathsf{R}_\bot$.}
	\label{fig:textbook-calculus}
\end{figure}
Figure~\ref{fig:textbook-calculus} shows the \Textbook classification rules for \EL from~\cite{BHLS-17} with a slight modification to rule $\mathsf{CR5}$ and the addition of a variant of~$\mathsf{R}_\bot$ from \Elk, in order to support \ELHb.
Here, all input and derived concept inclusions are required to be of one of the forms $A\sqsubseteq B$, $A_1\sqcap A_2\sqsubseteq B$, $A\sqsubseteq\exists r.B$ or $\exists r.A\sqsubseteq B$, where $A,A_1,A_2,B$ are concept names from~\Tmc, $\top$ or~$\bot$, which is again without loss of generality.
This calculus is correct in the same sense as the \Elk calculus above, but, instead of the $\mathsf{init}(A)$ statement, it is initialized with all axioms of the input TBox~\Tmc.

\begin{figure}[tb]
	\vskip-\smallskipamount
	\begin{framed}
		\centering
		\AXC{$C \sqsubseteq D$}
		\LeftLabel{$\mathrm{CR1}$}
		\RightLabel{\sideConColor{: $D \sqsubseteq E\in\Tmc$}}
		\UIC{$C\sqsubseteq E$}
		\DP
		\quad
		\AXC{$C \sqsubseteq D_1$\quad$C \sqsubseteq D_2$}
		\LeftLabel{$\mathrm{CR2}$}
		\RightLabel{\sideConColor{: $D_1 \sqcap D_2 \sqsubseteq E\in\Tmc$}}
		\UIC{$C\sqsubseteq E$}
		\DP
		\\[3ex]
		\AXC{$C \sqsubseteq D$}
		\LeftLabel{$\mathrm{CR3}$}
		\RightLabel{\sideConColor{: $D \sqsubseteq \exists r.E\in\Tmc$}}
		\UIC{$C \sqsubseteq \exists r.E$}
		\DP
		\quad
		\AXC{$C \sqsubseteq \exists r.D_1$\quad$D_1 \sqsubseteq D_2$}
		\LeftLabel{$\mathrm{CR4}$}
		\RightLabel{\sideConColor{: $\exists r.D_2 \sqsubseteq E\in\Tmc$}}
		\UIC{$C\sqsubseteq E$}
		\DP
		\\[3ex]
		\AXC{$C \sqsubseteq \exists r.D$\quad $D \sqsubseteq \bot$}
		\LeftLabel{$\mathrm{CR5}$}
		\UIC{$C \sqsubseteq \bot$}
		\DP
		\quad
		\AXC{$C \sqsubseteq \exists r.E$}
		\LeftLabel{$\mathrm{CR10}$}
		\RightLabel{\sideConColor{: $r \sqsubseteq s\in\Tmc$}}
		\UIC{$C \sqsubseteq \exists s.E$}
		\DP
		\\[3ex]
		\AXC{$C \sqsubseteq \exists r_1.D$\quad$D \sqsubseteq \exists r_2.E$}
		\LeftLabel{$\mathrm{CR11}$}
		\RightLabel{\sideConColor{: $r_1 \circ r_2 \sqsubseteq s\in\Tmc$}}
		\UIC{$C\sqsubseteq \exists s.E$}
		\DP
	\end{framed}
	\vskip-\smallskipamount
	\caption{\Envelope calculus~\cite{DBLP:conf/ijcai/BaaderBL05}, restricted to \ELpb.}
	\label{fig:envelope-calculus}
\end{figure}
Finally, Figure~\ref{fig:envelope-calculus} shows the \Envelope inference rules for \ELpb from~\cite{DBLP:conf/ijcai/BaaderBL05} (rules CR6--CR9 for nominals and concrete domains have been omitted).
For consistency with the other calculi, we have translated the statements \enquote{$D\in S(C)$} and \enquote{$(C,D)\in R(r)$} from the original paper into 
$C\sqsubseteq D$ and $C\sqsubseteq \exists r.D$, respectively.
Note that $C$ and $D$ in these statements must be concept names from~\Tmc, $\top$ or $\bot$.
This calculus also requires the concept and role inclusions in \Tmc to be in normal form, and is correct in the same sense as for the other calculi, but starting only from the tautologies $C\sqsubseteq C$ and $C\sqsubseteq \top$ for all concept names~$C$ from~\Tmc.

\mypar{Proofs}
Following the notion introduced in~\cite{DBLP:conf/lpar/AlrabbaaBBKK20}, a \emph{proof} of $\Tmc\models A\sqsubseteq B$, where $A$ and $B$ are concept names, is a finite, acyclic, 
directed \emph{hypergraph}, where each vertex~$v$ is labeled with an axiom $\ell(v)$.
Hyperedges represent sound inference steps and are of the form $(S, d)$, where~$S$ is a set of vertices and $d$ is a vertex s.t.\ $\{\ell(v)\mid v\in S\} \models \ell(d)$; \blue{note that the 
direction of the edge is from the vertices in $S$ to the vertex $d$}.
The leaf vertices (without incoming hyperedges) of a proof are labeled with axioms from \Tmc, and the unique root (without outgoing hyperedges) is labeled with $A\sqsubseteq B$.
Moreover, to obtain smaller proofs, they are also required to be \emph{non-redundant}, which means that the same axiom cannot have multiple subproofs: no two vertices can have the 
same label, every vertex can have at most one incoming hyperedge, and a vertex has no incoming hyperedge iff it is labeled by an axiom from~\Tmc.

Usually, hyperedges are additionally distinguished by a label, which corresponds to the rule name from a calculus.
However, for simplicity, in this paper, we dispense with hyperedges altogether and simply view proofs as directed acyclic graphs (with edges pointing towards the root) without rule names.
The children of a vertex~$v$ are then the premises of the (unique) inference step that was used to derive $\ell(v)$.
The label of a vertex without predecessors must then be either from~\Tmc or a tautology that can be derived using a rule without premises.
Moreover, since most proof visualizations~\cite{DBLP:conf/dlog/KazakovKS17,EVONNE} are based on tree-shaped proofs, in the following, we only consider the tree-shaped unravelings of these directed acyclic graphs
(see Figure~\ref{fig:example-proofs}).
This means that there can be multiple vertices with the same label, but then they must have isomorphic subproofs (subtrees).

To illustrate how to obtain proofs from the calculi above, we consider the small example TBox $\Tmc=\{A\sqsubseteq B,\ B\sqsubseteq\exists r.C,\ C\sqsubseteq D,\ \exists t.D\sqsubseteq E,\ r\sqsubseteq s,\ s\sqsubseteq t\}$ and the entailment $\Tmc\models A\sqsubseteq E$.
Since proofs are to be inspected by ontology engineers, they are restricted to %
DL axioms.
This means, however, that they cannot include side conditions and statements that are not DL axioms (\eg $\mathsf{init}(C)$).
Thus, we have to adapt the calculus rules before we can use them as inference steps in proofs.
\begin{figure}
  \centering
  \begin{tikzpicture}[grow'=up,edge from parent/.style={draw,latex-},sibling distance=5em,level distance=6ex]
    \node {$A\sqsubseteq E$}
      child {node {$A\sqsubseteq\exists t.D$}
        child {node {$A\sqsubseteq\exists r.C$}
          child {node {$A\sqsubseteq B$}}
          child {node {$B\sqsubseteq \exists r.C$}}
        }
        child {node {$C\sqsubseteq D$}}
        child {node {$r\sqsubseteq t$}
          child {node {$r\sqsubseteq s$}}
          child {node {$s\sqsubseteq t$}}
        }
      }
      child {node {$\exists t.D\sqsubseteq E$}}
    ;
  \end{tikzpicture}
  \begin{tikzpicture}[grow'=up,edge from parent/.style={draw,latex-},sibling distance=5em,level distance=6ex]
    \node {$A\sqsubseteq E$}
      child {node {$A\sqsubseteq\exists r.C$}
        child {node {$A\sqsubseteq B$}}
        child {node {$B\sqsubseteq \exists r.C$}}
      }
      child {node {$C\sqsubseteq D$}}
      child {node {$\exists t.D\sqsubseteq E$}}
      child {node {$r\sqsubseteq t$}
        child {node {$r\sqsubseteq s$}}
        child {node {$s\sqsubseteq t$}}
      }
    ;
  \end{tikzpicture}
  
  \medskip
  \begin{tikzpicture}[grow'=up,edge from parent/.style={draw,latex-},sibling distance=5em,level distance=6ex]
    \node {$A\sqsubseteq E$}
      child {node {$A\sqsubseteq\exists t.C$}
        child {node {$A\sqsubseteq\exists s.C$}
          child {node {$A\sqsubseteq\exists r.C$}
            child {node {$A\sqsubseteq B$}}
            child {node {$B\sqsubseteq \exists r.C$}}
          }
          child {node {$r\sqsubseteq s$}}
        }
        child {node {$s\sqsubseteq t$}}
      }
      child {node {$C\sqsubseteq D$}}
      child {node {$\exists t.D\sqsubseteq E$}}
    ;
  \end{tikzpicture}
  \caption{Example proofs based on \Elk (top left), \Textbook (top right), and \Envelope (bottom)}
  \label{fig:example-proofs}
\end{figure}
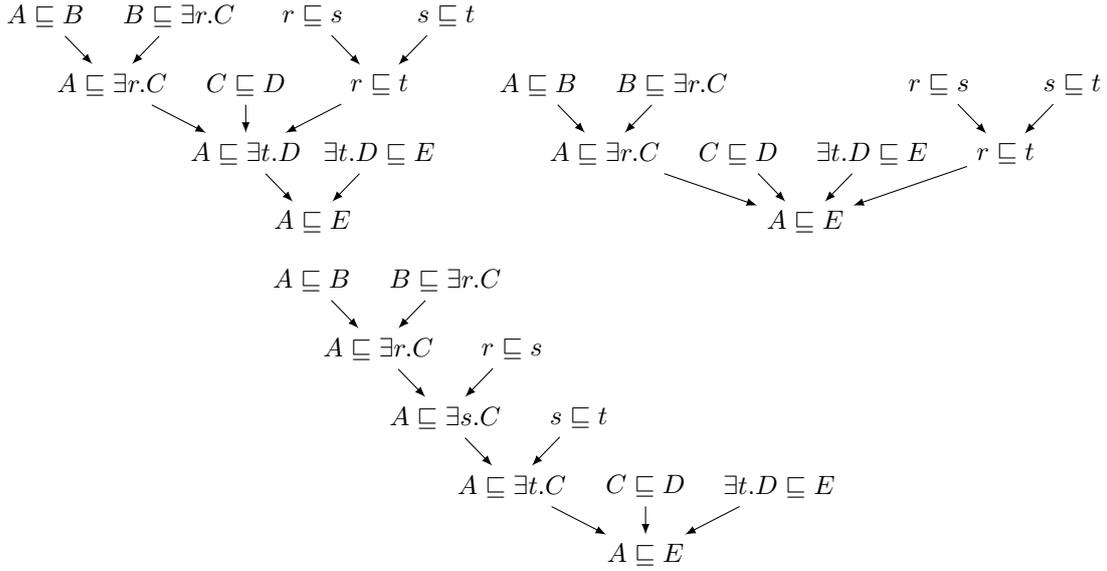

We treat side conditions of the form $C\sqsubseteq D\in\Tmc$ as ordinary premises $C\sqsubseteq D$, since they are DL axioms.
However, this does not mean that $\mathsf{R}_\sqsubseteq$ from \Elk is now equivalent to $\mathsf{CR3}$ from \Textbook, since the second premise in the former is still restricted to axioms from~\Tmc.

Similarly, we treat side conditions $r\sqsubseteq_\Tmc^* s$ as premises $r\sqsubseteq s$.
However, to obtain a valid proof, we also have to derive these axioms from the TBox.
For both the \Elk and the \Textbook calculus, we therefore add the following two rules:
\[
\AXC{\vphantom{$t\sqsubseteq t$}}
\RL{\sideConColor{: $t\in\NR$ occurs negatively in~\Tmc}}
\UIC{$t\sqsubseteq t$}
\DP
\quad
\AXC{$s\sqsubseteq t$}
\RL{\sideConColor{: $r\sqsubseteq s\in\Tmc$}}
\UIC{$r\sqsubseteq t$}
\DP
\]

Next, we replace statements of the form $C\xrightarrow{r}E$ in the \Elk calculus by $C\sqsubseteq\exists r.E$, which preserves soundness of the inference rules, and is similar to the treatment of \enquote{$(C,D)\in R(r)$} from the \Envelope calculus.
Additionally, we simply omit $\mathsf{init}(C)$ statements, which means that $\mathsf{R}_0,\mathsf{R}_1$ from \Elk become similar to $\mathsf{CR1},\mathsf{CR2}$ from \Textbook, but they can only appear in a proof if $\mathsf{init}(C)$ was actually derived by the original rules.
We also add $\mathsf{CR1}$ and $\mathsf{CR2}$ to the \Envelope calculus to express the initialization step, which does not correspond to explicit rules in~\cite{DBLP:conf/ijcai/BaaderBL05}.
For the \Elk and \Envelope calculi, which are not explicitly initialized with the TBox axioms, this results in proofs such as the following:
\[
\AXC{}
\UIC{$A\sqsubseteq A$}
\AXC{$A\sqsubseteq B$}
\BIC{$A\sqsubseteq B$}
\DP
\]
However, this violates non-redundancy since $A\sqsubseteq B$ has multiple non-isomorphic proofs.
In such cases, we keep only a subproof of minimal size (in the example, only the leaf $A\sqsubseteq B$).
This also applies to other inference rules, for example, $\mathsf{R}_{\exists^-}$, which now trivially derives $C\sqsubseteq\exists r.E$ from $C\sqsubseteq\exists r.E$.
Thus, due to the above simplification, this rule will never appear in a proof.

We also omit the side conditions of the form \enquote{$X$ occurs negatively in~\Tmc}.
To represent them in our proofs, we could add the axiom in which $X$ occurs negatively as a premise; however, this axiom is not logically necessary for the inference step, and may be confusing because it has no connection to the inference other than the occurrence of~$X$.
For example, consider the following application of~$\mathsf{R}_\exists^+$, which also requires that $\exists t.D$ occurs negatively in~\Tmc:
\[
\AXC{$A\sqsubseteq\exists r.C$}
\AXC{$C\sqsubseteq D$}
\AXC{$r\sqsubseteq t$}
\TIC{$A\sqsubseteq\exists t.D$}
\DP
\]
We could add $\exists t.D\sqsubseteq E$ as a premise to express the side condition, but this axiom is actually not necessary to derive the conclusion $A\sqsubseteq\exists t.D$.
Moreover, $\exists t.D\sqsubseteq E$ also occurs in the subsequent application of $\mathsf{R}_\sqsubseteq$, and thus it may be additionally confusing if it occurs in two consecutive proof steps, but is only necessary for one of them.

The proofs resulting from all these considerations can be seen in Figure~\ref{fig:example-proofs}, where we draw them as trees to emphasize their shape, which we want to analyze in the following.
We can see that, despite the transformations and simplifications we applied, the three calculi can still result in substantially different proofs, even for such a simple entailment.

\section{Measures}
\label{sec:measures}
We evaluate the shape of proofs by several measures from the literature, which reflect different aspects of how ontology engineers can inspect a proof in different representations (see Figure~\ref{fig:shapes}).

\begin{figure}[tb]
  \centering
  ~
  \hfill
  \begin{tikzpicture}[axiom/.style={draw,rectangle,rounded corners=3pt,minimum width=6em,minimum height=3ex},scale=0.5,transform shape,node distance=2em and 1.5em,baseline=(j.base)]
    \node (a) {};
    \node[below right=of a] (b) {};
    \node[below right=of b] (d) {};
    \node[below=of d] (e) {};
    \node[below right=of e] (i) {};
    \node[below=of i] (j) {};
    \node[below=13em of b] (c) {};
    \node[below right=of c] (f) {};
    \node[below=of f] (g) {};
    \node[below=of g] (h) {};

    \node[axiom,fill=gray, right=-1em of a] {};
    \node[axiom, right=-1em of b] {};
    \node[axiom, right=-1em of c] {};
    \node[axiom, right=-1em of d] {};
    \node[axiom, right=-1em of e] {};
    \node[axiom, right=-1em of f] {};
    \node[axiom, right=-1em of g] {};
    \node[axiom, right=-1em of h] {};
    \node[axiom, right=-1em of i] {};
    \node[axiom, right=-1em of j] {};

    \begin{scope}[-latex,shorten >=3pt,shorten <=3pt,rounded corners=3pt]
    \draw (b) -| (a) ;
    \draw (d) -| (b) ;
    \draw (e) -| (b) ;
    \draw (i) -| (e) ;
    \draw (j) -| (e) ;
    \draw (c) -| (a) ;
    \draw (f) -| (c) ;
    \draw (g) -| (c) ;
    \draw (h) -| (c) ;
    \end{scope}

    \path[every edge/.style={gray,draw,|-|}]
      (-0.6,0.25) edge node[sloped,yshift=-0.9em] {\Large size} (-0.6,-8.8)
      (-0.2,-9.2) edge node[sloped,yshift=-0.9em] {\Large depth} (4.2,-9.2)
    ;
  \end{tikzpicture}
  \hfill
  \begin{tikzpicture}[axiom/.style={draw,rectangle,rounded corners=3pt,minimum width=6em,minimum height=3ex},scale=0.5,transform shape,node distance=2em and 1.5em,baseline=(i.base)]
    \node (a) {};
    \node[above=of a] (b) {};
    \node[above=of b] (d) {};
    \node[above=of d] (e) {};
    \node[above=of e] (i) {};
    \node[above=of i] (j) {};
    \node[above=of j] (c) {};
    \node[above=of c] (f) {};
    \node[above=of f] (g) {};
    \node[above=of g] (h) {};

    \node[axiom,fill=gray, right=-1em of a] {};
    \node[axiom, right=-1em of b] {};
    \node[axiom, right=-1em of c] {};
    \node[axiom, right=-1em of d] {};
    \node[axiom, right=-1em of e] {};
    \node[axiom, right=-1em of f] {};
    \node[axiom, right=-1em of g] {};
    \node[axiom, right=-1em of h] {};
    \node[axiom, right=-1em of i] {};
    \node[axiom, right=-1em of j] {};

    \begin{scope}[]
    \path[-latex,shorten >=3pt,shorten <=3pt,bend right=45]
      (b.west) edge (a.west) 
      (d.west) edge (b.west)
      (e.west) edge (b.west)
      (i.west) edge (e.west)
      (j.west) edge (e.west)
      (c.west) edge (a.west)
      (f.west) edge (c.west)
      (g.west) edge (c.west)
      (h.west) edge (c.west)
    ;
    \end{scope}

    \path
      (-1.8,3.3) edge[draw=gray,|-|] (-0.2,3.3)
      (-1.8,-0.6) edge[draw=none] node[gray,sloped,yshift=-0.9em] {\Large cutwidth} (-0.2,-0.6)
    ;
  \end{tikzpicture}
  \hfill
  \begin{minipage}[c]{0.5\textwidth}
  \begin{tikzpicture}[every node/.style={draw,rectangle,rounded corners=3pt,minimum width=6em,minimum height=3ex},scale=0.5,transform shape,level distance=2.7em,baseline=(i.base)]
    \node[fill=gray] (a) {} [grow'=up,latex-,sibling distance=22em]
      child {node (b) {} [sibling distance=10em]
        child {node (d) {}}
        child {node (e) {} [sibling distance=7em]
          child {node (i) {}}
          child {node (j) {}}
        }
      }
      child {node (c) {} [sibling distance=7em]
        child {node (f) {}}
        child {node (g) {}}
        child {node (h) {}}
      }
    ;

    \path[every node/.style={draw=none},every edge/.style={gray,draw,|-|}]
      (-7.1,3.1) edge node[sloped,yshift=-0.9em] {\Large depth} (-7.1,-0.2)
      (-6.7,-0.6) edge node[sloped,yshift=-0.9em] {\Large justification size} (7.3,-0.6)
    ;
  \end{tikzpicture}
  \end{minipage}
  \hfill
  ~
  \caption{Different visual representations of a proof: nested list (left), linear (middle), proof tree (right).} %
  \label{fig:shapes}
\end{figure}
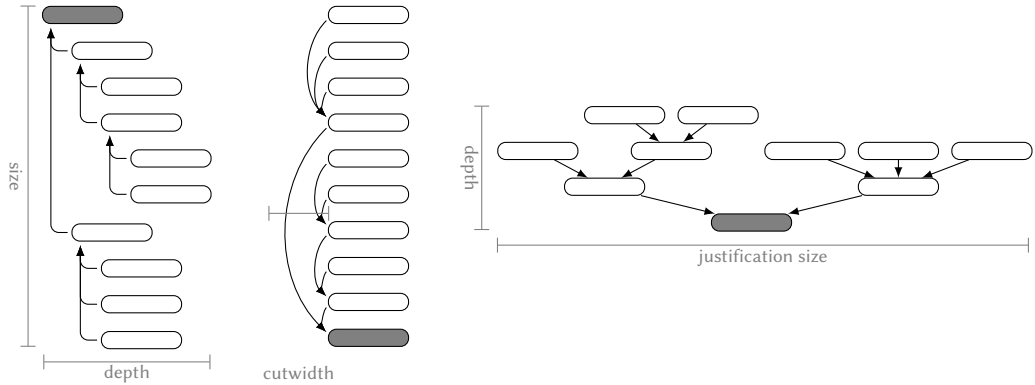

\mypar{Size}
The \emph{size} of a proof~\cite{DBLP:conf/lpar/AlrabbaaBBKK20} is the number of occurrences of axioms in the proof tree.
For example, if ontology engineers need to inspect the whole proof because they are unfamiliar with the ontology, the size of the proof influences how much time they will need to inspect it.
The size is proportional to the height of a linear representation of the proof or a nested list visualization as in a file browser, \eg in the \emph{proof explanation plugin} for~\Protege~\cite{DBLP:conf/dlog/KazakovKS17}.

\mypar{Depth}
The \emph{depth} of a proof~\cite{DBLP:conf/cade/AlrabbaaBBKK21} is the length of the longest path from a leaf to the root.
This is related to the task of finding an error in the ontology based on an erroneous inference.
In this case, ontology engineers would rather inspect the proof only along the \enquote{most suspicious-looking} branch (or a few such branches), and thus the time they may need in the worst case depends on the length of the longest branch.
The depth is also proportional to the height of a more traditional proof-tree representation, as used in \Evonne~\cite{EVONNE}.
The width of such a tree visualization, assuming that adjacent subtrees are \emph{non-overlapping} in the sense that nodes from one subtree are not positioned vertically above nodes from another subtree, can be expressed as a function of the number of leafs in the proof, which is also called its \emph{justification size}~\cite{DBLP:conf/dlog/AlrabbaaBBKK20}.
However, we do not consider the justification size in this paper, since it does not depend on the reasoning calculus.

\mypar{Cutwidth}
The \emph{(directed) cutwidth} of a graph tries to measure how linear it is~\cite{DBLP:journals/jcss/BodlaenderFT09}.
Given a linear vertical arrangement of a proof's nodes~\cite{EVONNE}, the cutwidth is the maximal number of edges that would be affected when cutting the graph horizontally between any two consecutive nodes (see Figure~\ref{fig:shapes}).
The cutwidth of the graph is then the minimum of the cutwidths of all possible such linear arrangements.
Hence, it is proportional to the maximal number of intermediate axioms an ontology engineer needs to keep in memory when reading the proof in such an (optimal) linear representation from top to bottom.

\mypar{Bushiness Score}
We developed another score to measure how \enquote{bushy} (non-linear) a proof is.
It is computed as the ratio between the size of the proof and its depth ($+1$ for the root).
Hence, it can be interpreted as the average number of vertices per level (see Figure~\ref{fig:shapes}).
A completely linear proof in which all inference steps have only one premise results in a bushiness score of~$1$ (not bushy at all), and the full binary tree with five levels gets a score of $\frac{31}{5}=6.2$ (very bushy).

\mypar{Step Complexity}
Finally, we consider a measure based on the cognitive complexity of inference steps, which has been proposed in the context of justifications, and has been evaluated in user studies~\cite{DBLP:journals/kbs/HorridgeBPS13}.
It reflects multiple aspects of the syntactical structure of the argument, such as the depth of involved concepts, how many constructors and axiom types are used, or whether it uses the triviality of a concept name (\ie being equivalent to $\bot$ or $\top$).
Here, we consider the average of the step complexities of the individual inference steps in the proof.

\section{Directed Cutwidth on Trees}\label{sec:cutwidth}

Computing cutwidth is \NP-complete in general \cite{DBLP:journals/jcss/BodlaenderFT09}, and the general-purpose implementations we tried could not compute
the metric for our complete evaluation set in a feasible time.
The problem becomes polynomial on trees \cite{DBLP:journals/jacm/Yannakakis85}, but the algorithm is complicated, and we are not
aware of any implementation. \emph{Directed} cutwidth on trees, however, admits a far simpler algorithm, which we now introduce and prove to be correct,
since we could not find any publication that establishes this result.

Consider a tree $T=\tuple{V,E}$ with vertices $V$ and directed edges $E$ pointing towards the leafs.\footnote{This edge direction is more intuitive for this section, but
our results are readily applied to proof trees with edges oriented the other way (see Figures~\ref{fig:example-proofs} and~\ref{fig:shapes}). Indeed, directed cutwidth is
the same for the dual ordering.}
A \emph{serialization} of $T$ is a word $S\in V^*$ that contains each vertex $v\in V$ at a unique position
$v_S\in\{1,\ldots,|V|\}$ such that $\tuple{v,w}\in E$ implies $v_S<w_S$.
For $i\in\{1,\ldots,|V|-1\}$, the number of edges cut in the $i$th gap between vertices in $S$ is
$\cut{i}=|\{\tuple{v,w}\in E\mid v_S\leq i< w_S\}|$.
The \emph{cutwidth} $\dcw{S}$ of $S$ is $\max_{1\leq i<|V|} \cut{i}$, and the \emph{directed cutwidth} $\dcw{T}$ of $T$ is the minimal cutwidth of any serialization of~$T$.
The out-degree of any vertex in $T$ is a lower bound for its directed cutwidth, though it can be higher (e.g., for a full binary tree $T$, $\dcw{T}$ is depth${}+1$).
We omit \emph{directed} below, since we consider no other kind of cutwidth on trees.

\begin{definition}\label{def_std_serial}
The \emph{standard serialization} $S(T)$ of a tree $T=\tuple{V,E}$ is defined inductively.
If $V=\{v\}$, then $S(T):=v$.
If $|V|>1$, then let $r$ be the root of $T$ and let $C_1,\ldots,C_\ell$ be its direct subtrees, ordered such that 
$i<j$ implies $\dcw{S(C_i)}\leq \dcw{S(C_j)}$; then $S(T):=r S(C_1) \cdots S(C_\ell)$.
\end{definition}

For non-singleton trees $T$, each of the sub-sequences $S(C_i)$ has $\ell-i$ ``overarching'' edges from the root
to the roots of later $C_j$, $j>i$ (see also Figure~\ref{fig:shapes}). %
Therefore, if the root of $T$ has $\ell$ children, $\dcw{S(T)}=\max(\{\ell\}\cup\{\dcw{S(C_i)}+\ell-i\mid 1\leq i\leq \ell\})$.
It is easy to compute $S(T)$ and $\dcw{S(T)}$ in a bottom-up fashion.
To do this in small steps, after computing $S(C_i)$ and $\dcw{S(C_i)}$ for all child trees $C_i$ of a vertex $r$,
we can add the children $C_i$ to $r$ one by one, in an order of non-decreasing cutwidth.
The following lemma is the key to showing that this recursive procedure yields, for each
partial subtree (with only some child trees added yet), the exact cutwidth.
Its proof can be found in the appendix.

\begin{restatable}{lemma}{LemDcwInductionStep}\label{lemma_dcw_induction_step}
For $i\in\{1,2\}$, let $T_i$ be a tree with root $r_i$ and cutwidth $w_i=\dcw{T_i}=\dcw{S_i}$ for an optimal serialization
$S_i$.
Assume that $C_1,\ldots,C_\ell$ are the direct child trees below $r_1$ in an order of non-decreasing cutwidth,
and that $w_2\geq \dcw{C_i}$ for all $1\leq i\leq \ell$.
Then the tree $T$ obtained from $T_1$ by adding $T_2$ as an additional direct child tree below $r_1$ has
cutwidth $\dcw{T}=\max(w_1+1, w_2)$, and an optimal serialization is $S_1 S_2$.
\end{restatable}

Now we can perform an easy induction over the structure of $T$ to show that 
the recursive computation of $\dcw{S(T)}$ produces $\dcw{T}$.
The induction base are single-node trees (leafs), and the step is Lemma~\ref{lemma_dcw_induction_step}.

\begin{theorem}\label{theo_dcw_computation_correct}
For trees $T$, $\dcw{T}=\dcw{S(T)}$, which can thus be computed in polynomial time.
\end{theorem}

\section{Implementation}
\label{sec:implementation}

We now describe the encoding of the calculi into existential rules with stratified negation in \Nemo syntax, and the subsequent translation of rule-based reasoning traces into DL proofs.
We have implemented this in the DL explanation library \Evee.\footnote{See \url{https://github.com/de-tu-dresden-inf-lat/evee/tree/nemo-extractor/evee-nemo-proof-extractor}}

\subsection{Existential Rules}
\label{sec:nemo-rules}

The implementation of the \Elk calculus in existential rules with stratified negation is split into three stages~\cite{DBLP:conf/kr/Ivliev0MSK24}.
The first stage consists of RDF import and normalization.
Here, the input OWL ontology is translated into facts over several predicates with the prefix \texttt{nf}.
These facts serve as the input for the calculus, \eg
\texttt{nf:subClassOf(A,B)}
for $A\sqsubseteq B\in \Tmc$.\footnote{More precisely, concept names are identified by IRIs, \eg \texttt{<http://example.org/iriA>}, but we omit them here.}
The normalization takes advantage of the OWL RDF encoding, which already contains a blank node for each complex concept, which can be used as an auxiliary concept name identifying that concept.
For example \texttt{nf:exists(?C,?R,?D)} states that \texttt{?C} is an auxiliary concept name representing the existential restriction with role \texttt{?R} and filler concept \texttt{?D}.
All such auxiliary concept names are marked by the predicate $\texttt{auxClass}$, which allows identifying the named concepts occurring in the input ontology using the rule
\[
\texttt{nf:isMainClass(?X) :- class(?X), $\sim$auxClass(?X) .}
\]
Auxiliary role names are treated similarly.
Since there can be multiple blank nodes denoting the same complex concept, there are additional existential rules that create unique representatives for these concepts, in order to avoid redundant computations.
Additionally, the normalization stage precomputes the role hierarchy $\sqsubseteq_\Tmc^*$ in the predicate \texttt{nf:subProp}.

The second stage consists of the actual inference rules of the \Elk calculus, which derive additional facts over predicates with the prefix \texttt{inf}.
Due to the normalization, the implementation of these rules is straightforward and close to the original calculus, for example for the rule $\mathsf{R}_\sqsubseteq$:
\[
\texttt{inf:subClassOf(?C,?E) :- inf:subClassOf(?C,?D), nf:subClassOf(?D,?E) .}
\]

In the third stage, the interesting entailments, namely subsumptions between concept names, are extracted with the help of the \texttt{inf:subClassOf} and \texttt{nf:isMainClass} predicates.

\mypar{Modifications.}
We slightly adapted these rules for our purposes.
All modifications are confined to the normalization and extraction; the inference rules of the calculus are not affected.

To the normalization stage, we added an inference rule that derives $\bot \sqsubseteq A$  for each concept name $A$ in the input, which ensures that $C \sqsubseteq A$ is inferred for any concept $C$ with $\Tmc \models C \sqsubseteq \bot$. 
Such reasoning with the $\bot$-concept is handled outside the calculus of \citet{DBLP:journals/jar/KazakovKS14}, we explicitly included it to obtain a complete reasoner for classification.

Next, we updated the precomputation of the role hierarchy, which is recursively computed in a top-down manner, starting with the largest role in a set of connected role inclusions, rather than in a bottom-up fashion (as in the original rules).
\begin{verbatim}
    directSubProp(?R,?S) :- TRIPLE(?R,rdfs:subPropertyOf,?S) .
    nf:subProp(?R,?R) :- nf:exists(?C,?R,?D), nf:isSubClass(?C) .
    nf:subProp(?R,?T) :- directSubProp(?R,?S), nf:subProp(?S,?T) .
\end{verbatim}
For example, starting from $r \sqsubseteq s$, $s \sqsubseteq t$, $t \sqsubseteq u$, we first derive $s \sqsubseteq u$ and then $r \sqsubseteq u$. This modification is necessary, because the original implementation did not compute the role hierarchy correctly. It started from roles~$r$ occurring negatively in~\Tmc (using \texttt{nf:isSubClass} as above), but only computed the role hierarchy \emph{above}~$r$, the opposite of what is needed by~$\mathsf{R}_\exists^+$.

We also added support for transitivity and domain axioms by introducing two rules that translate $\mathsf{domain}(r,C)$ into $\exists r.\top \sqsubseteq C$ and $\mathsf{trans}(r)$ into $r \circ r \sqsubseteq r$.
Moreover, we support the current OWL~2 encoding of role chains, \ie \texttt{owl:propertyChainAxiom}, including role chains with more than two role names.
Finally, we added the possibility to prove equivalence axioms $C \equiv D$ directly from $C \sqsubseteq D$ and $D \sqsubseteq C$.

\mypar{Other Calculi.}
The same normalization is also utilized in the implementation of the \Textbook and \Envelope calculi, but for those we added the additional normalization predicates \texttt{nf:subClassConj(?C,?D,?E)} for $C\sqcap D\sqsubseteq E$, \texttt{nf:subClassEx(?R,?C,?D)} for $\exists r.C\sqsubseteq D$, and \texttt{nf:supClassEx(?C,?R,?D)} for $C\sqsubseteq\exists r.D$, according to the normal forms expected by these calculi.
The main inference rules of both calculi are implemented as expected, for example
\begin{align*}
\texttt{in}&\texttt{f:subClassOf(?C,?F) :- inf:supClassEx(?C,?R,?D),} \\
&\texttt{inf:subClassOf(?D,?E), nf:subProp(?R,?S), inf:subClassEx(?S,?E,?F) .}
\end{align*}
for $\mathsf{CR5'}$ from the \Textbook calculus, and
\[
\texttt{R(?C,?R,?E) :- S(?C,?D), nf:supClassEx(?D,?R,?E)}
\]
for CR3 from the \Envelope calculus, where we used the original names $R$ and $S$~\cite{DBLP:conf/ijcai/BaaderBL05}.

\subsection{Translation to OWL Proofs}

\Nemo allows us to obtain a \emph{trace} of how a particular entailment, \eg \texttt{inf:subClassOf(A,B)}, was obtained using the rules.
Such a trace is similar to a DL proof, but the vertices are labeled with ground facts instead of DL axioms.
In the Java class \texttt{NemoProofParser}, we implemented a translation from such traces into valid DL proofs, which reads a Nemo trace in JSON format into an \texttt{IProof<String>} object, where each \texttt{String} represents a fact, and outputs an \texttt{IProof<OWLAxiom>} object\footnote{\texttt{OWLAxiom} is part of the Java OWL API~\cite{DBLP:journals/semweb/HorridgeB11}.} that can be further processed or again saved in JSON format.

The abstract class \texttt{AbstractAtomParser}, instantiated
for each of the three calculi,
translates facts into DL axioms,
\eg \texttt{inf:subClassOf(A,B)} into $A\sqsubseteq B$ and \texttt{nf:supClassEx(A,r,B)} into $A\sqsubseteq\exists r.B$.
Any fact that does not have a corresponding DL axiom, \eg \texttt{inf:init(A)}, is translated into $\bot \sqsubseteq \top$, which indicates that the atom is skipped and will not appear in the final proof (see Section~\ref{sec:preliminaries}).
In particular, traces contain many normalization steps that use various auxiliary predicates, which do not show up in the final DL proof.

Auxiliary concept names that are represented by blank nodes (either from the RDF encoding or from existential rules) need special treatment, since they would not make sense in the translated proof.
Instead, we replace them by the concepts they denote.
Since a fact containing a blank node, such as \texttt{inf:subClassOf(A,\_:61)}, does not contain information about the complex concept the blank node identifies, we need to collect other relevant facts from the trace.
For example, \texttt{nf:exists(\_:61,t,D)} encodes that \texttt{\_:16} stands for $\exists t.D$.
If \texttt{D} is also a blank node, then we have to recursively consider more facts until we can construct the final concept.

The only place where this approach does not work is for complex role inclusions with auxiliary role names.
For example, if the role inclusion $r\circ s\circ t \sqsubseteq u$ is normalized into \texttt{nf:subPropChain(r,s,\_:5)} and \texttt{nf:subPropChain(\_:5,t,u)}, we cannot simply replace \texttt{\_:5} by $r\circ s$, since that would result in $r\circ s\sqsubseteq r\circ s$, which is not expressible in OWL due to the role composition on the right-hand side.
However, this axiom is a tautology, and thus we can omit it without affecting the correctness of the inference steps.
During reasoning, auxiliary roles can also appear in existential restrictions, \eg in \texttt{inf:supClassEx(C,\_:5,D)}, which we translate into two (or more) nested existential restrictions: $C\sqsubseteq\exists r.\exists s.D$.

Lastly, after translating the inference steps in the described manner, we minimize the size of the resulting OWL proof using the \texttt{MinimalProofExtractor} class, which eliminates redundant inferences and unnecessary tautologies~\cite{DBLP:conf/cade/AlrabbaaBBKK21}.
In particular, this removes all occurrences of $\bot \sqsubseteq \top$.
As a result, we obtain a correct and complete DL proof from the \Nemo reasoning trace.

\section{Evaluation}
\label{sec:evaluation}

We compared the proofs resulting from the three calculi on a benchmark of 1,573 reasoning tasks that were extracted from the ORE 2015 ontologies~\cite{DBLP:journals/jar/ParsiaMGGS17,DBLP:conf/lpar/AlrabbaaBBKK20}.
Each reasoning task consists of a justification and the entailed axiom, which avoids the overhead of having to deal with a large ontology and lets us focus on the structure of the inference steps used to prove the entailments.
The benchmark covers all types of entailments that hold in the ORE ontologies, modulo renaming of concept and role names (and some timeouts).
The resulting proofs and measurements can be found on Zenodo~\cite{alrabbaa_2025_16320822}, and in the appendix there are detailed graphs for pairwise comparisons of the calculi.

From the structure of the calculi, we expected proofs of the \Textbook calculus to be less linear and shallower, \ie have larger directed cutwidth, larger bushiness score and smaller depth, because it does not restrict any premises of the inference rules to be from the input TBox (see also Figure~\ref{fig:shapes}), and therefore allows for more balanced proof trees.
This was confirmed in the experiments, with the directed cutwidth of \Textbook proofs being higher in 1,486 and 1,381 cases compared to \Elk and \Envelope, respectively, and never lower.
Similarly, the bushiness score is higher for 1,534 and 1,512 proofs, and lower in only 12 and 23 cases, respectively.
Conversely, the depth is lower for 1,503 proofs, and higher in only 8 cases, compared to both other calculi.
We conjecture that the depth of the \Textbook proofs is approximately logarithmic in the depth of the corresponding \Elk or \Envelope proofs, and that the relationship for directed cutwidth and bushiness score is exponential.
On the same three measures, the \Elk and \Envelope proofs behave very similarly, with \Envelope proofs having slightly higher directed cutwidth and bushiness score, but nearly identical depth, compared to the \Elk proofs.

We also compared directed cutwidth and bushiness score, since they both try to measure how linear proofs are. We observe that there is indeed a correlation between them; however, whereas directed cutwidth is always a natural number, bushiness score allows a more fine-grained comparison of proofs, and also tends to increase faster than the directed cutwidth.

For the remaining measures of size and average step complexity, the \Elk calculus is the outlier, with both lower size (in 1,025 and 584 cases compared to \Envelope and \Textbook, 
respectively) and lower average step complexity (1,281 and 1,062 proofs, respectively).
Here, \Envelope and \Textbook obtain very similar results, with slightly lower values for \Textbook.

\blue{There is no clear winner across all the measures we considered. However, depending on specific use cases, proofs generated using certain calculi may be more 
preferable.
	Overall, proofs generated with the \Elk calculus seem to be the better choice for providing explanations of entailments, since they are smaller and rely on simpler inferences.
	Additionally, \Elk proofs have lower cutwidth and bushiness scores, indicating that axioms are used as soon as possible in the proof, which can be an advantage when reading the proofs 
	in a linear format (see Figure~\ref{fig:shapes}).
  Conversely, the low depth of \Textbook proofs may be a better option for other types of visualizations, since it can reduce the amount of vertical or horizontal scrolling required, which also allows users to inspect the proof in a more linear manner.
	Our experiment do not show a specific advantage for proofs generated using the \Envelope calculus over the other two.
}

\subsection{Limitations}

The benchmarks are not fully representative of general \ELpb proofs, mainly for two reasons.
First, the benchmark does not contain all possible justifications and entailments from the ORE 2015 ontologies, since for some cases it was too costly to compute all justifications~\cite{DBLP:conf/lpar/AlrabbaaBBKK20}.
Second, \Nemo always computes only one trace, even if there are different combinations of inference steps that result in the same axiom.
Since this is done by a specific algorithm in \Nemo, there is a systematic bias in the resulting proofs.
For example, for the \Textbook calculus, instead of very bushy proofs, \Nemo could also have returned more linear proofs, as for the other calculi.
In this case, \Nemo traces seem to be biased towards smaller depth, which allows us to observe the differences between the calculi and confirms our initial intuitions.
\blue{In future work, we plan to reevaluate this once \Nemo is extended to consider all possible traces rather than just one.}%

\section{Conclusion}
\label{sec:conclusion}

We compared the proofs obtained from three reasoning calculi for the \EL family of DLs.
This was facilitated by \Nemo, a powerful rule engine that allowed us to quickly implement the calculi without having to develop a dedicated reasoner.
As expected, the \Textbook calculus~\cite{BHLS-17} indeed produces more bushy and more shallow proofs, and it turns out that the \Elk calculus~\cite{DBLP:journals/jar/KazakovKS14} generally yields smaller proofs whose inference steps are on average less complex~\cite{DBLP:journals/kbs/HorridgeBPS13}.
This enables us to choose specific calculi for different purposes, \eg to show proofs in a visualization format where the screen space is restricted either horizontally or vertically, or when the goal is to be able to understand individual inference steps more quickly.
In the future, we want to apply this method to consequence-based calculi for more expressive logics, \eg \ALC and beyond.

\begin{acknowledgments}
  This work was supported by DFG in grant 389792660 (TRR~248: \href{https://perspicuous-computing.science}{CPEC}, \href{https://perspicuous-computing.science}{https://perspicuous-computing.science}),
  and by BMFTR and DAAD in project 57616814 (\href{https://secai.org/}{SECAI}, \href{https://secai.org/}{School of Embedded and Composite AI}).
\end{acknowledgments}

\pagebreak
\appendix
 
\section{Proof of Lemma~\ref{lemma_dcw_induction_step}}

\LemDcwInductionStep*
\begin{proof}
We observe that $\dcw{S_1 S_2}=\max(w_1+1, w_2)$, so it remains to show that this is optimal.
Suppose for a contradiction that $\dcw{S}=\dcw{T}<\max(w_1+1, w_2)$ for some serialization $S$.
Since $T$ contains all edges of either $T_i$, $\dcw{T}\geq w_1$ and $\dcw{T}\geq w_2$,
so $\dcw{T}=w_1$.

For $1\leq i\leq\ell$, let $c_i=\dcw{C_i}$, let $S^c_i$ be a serialization with $c_i=\dcw{S^c_i}$,
and let $\opfont{end}_i$ be the maximal position of a vertex in $C_i$ in $S$.
Likewise, let $\opfont{end}$ be the maximal position of a vertex of $T_2$ in $S$.
Now $k=|\{j\mid 1\leq j\leq \ell, \opfont{end} < \opfont{end}_j \}|$ is the number of child trees $C_j$ 
ending after $T_2$ in $S$.
We find $\dcw{T}\geq w_2+k$, since every gap in $S$ from $1$ to $\opfont{end}_j-1$ cuts at least one edge 
of the subtree $r_1\to C_j$ of $T$, and the sub-sequence for $T_2$ in~$S$ must have a gap that cuts $w_2$ edges of $T_2$.

Now consider the serialization $S_1'=r_1 S^c_1\cdots S^c_\ell$. $S_1'$ is analogous to the standard serialization
of $T_1$ but using arbitrary optimal (possibly non-standard) sub-serializations for $C_i$,
and, analogously, we find $\dcw{S_1'}=\max(\{\ell\}\cup\{\dcw{C_i}+\ell-i\mid 1\leq i\leq \ell\})$.
We abbreviate $w_1^s=\dcw{S_1'}$, and note that $w_1^s\geq w_1=\dcw{T}$.

Case (A): If $w_1^s=\ell$, since $w_1\geq\ell$ and $\dcw{T}\geq\ell+1$ (cutwidth $\geq$ out-degree), we get $w_1=\ell$ and $\dcw{T}\geq w_1+1$, contradicting $\dcw{T}=w_1$.

Case (B): There is $m\in\{1,\ldots,\ell\}$ such that $w_1^s=\dcw{C_m}+\ell-m$.
By the ordering, we have $\dcw{C_j}\geq\dcw{C_m}$ for all $m\leq j\leq\ell$;
the number of these \emph{high children} is $h=\ell-m+1$.
Since $w_1^s\geq w_1$ and $w_1=\dcw{T}\geq w_2+k$, we get $\dcw{C_m}+\ell-m\geq w_2+k$.
But $w_2\geq\dcw{C_m}$ by the precondition of the lemma, hence $w_2+\ell-m\geq w_2+k$,
and so $\ell-m\geq k$, i.e., $h\geq k+1$.

Therefore, $h>k$, i.e., at least one high child $C_i$ in $T_1$ has all of its vertices before position $\opfont{end}$ in $S$.
Thus, there is some gap in the range of $C_i$ in~$S$ that cuts $\dcw{C_i}$ edges from $C_i$, (at least) $k$ edges from subtrees $r_1\to C_j$ with $\opfont{end}_j>\opfont{end}>\opfont{end}_i$, and (at least) one edge from $r_1\to T_2$, so in total $\dcw{C_i}+k+1$ edges.
Hence, $\dcw{T}\geq \dcw{C_i}+k+1\geq \dcw{C_m}+k+1$ since $C_i$ is a high child.
But then $\dcw{C_m}+k+1\leq \dcw{T}= w_1\leq w_1^s = \dcw{C_m}+\ell-m$, therefore $k+1\leq \ell-m$, i.e., $k+2\leq h$.

By analogous reasoning, we find two high children $C_i$ and $C_j$ in $T_1$ both having all of their vertices before $\opfont{end}$ in $S$.
Without loss of generality, assume that $C_i$ has all of its vertices before $\opfont{end}_j$. Then $C_i$ has at least $k+2$ overarching edges
belonging to later child trees and $T_2$, so $\dcw{T}\geq \dcw{C_m}+k+2$. By the same steps as above, we find $k+3\leq h$.
Continuing this argument, we can show that $h-k$ must be arbitrarily large, which cannot be (in fact, $k,h\leq \ell$). Contradiction.
\end{proof}

\section{Graphs for the Evaluation}

These graphs compare each pair of calculi on each of the measures (Figures~\ref{fig:cutwidth}--\ref{fig:tree-size}), and compare directed cutwidth and bushiness score for each of the calculi (Figure~\ref{fig:bushinessToCutwidth}).
Each red circle corresponds to one of the 1,573 reasoning tasks from the benchmark~\cite{DBLP:conf/lpar/AlrabbaaBBKK20}.

\begin{figure}[ht]
	\centering
	\includegraphics[scale=1.3]{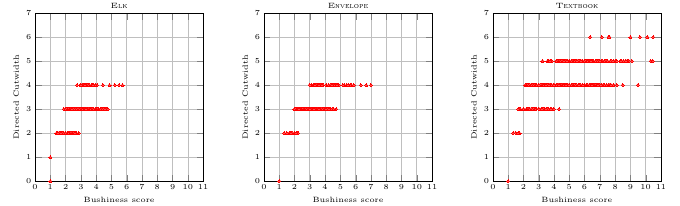}
	\caption{Relation between bushiness score and directed cutwidth}
	\label{fig:bushinessToCutwidth}
\end{figure}

\begin{figure}[ht]
	\centering
	\includegraphics[scale=1.3]{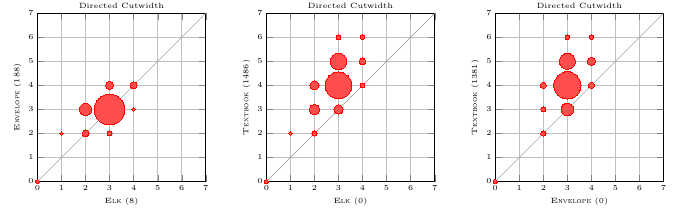}
	\caption{Directed cutwidth. Circle area is scaled proportionally to the number of corresponding proofs.}
	\label{fig:cutwidth}
\end{figure}

\begin{figure}[ht]
	\centering
	\includegraphics[scale=1.3]{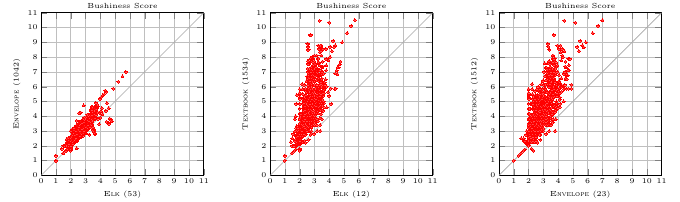}
	\caption{Bushiness score}
	\label{fig:bushiness}
\end{figure}

\begin{figure}[ht]
	\centering
	\includegraphics[scale=1.3]{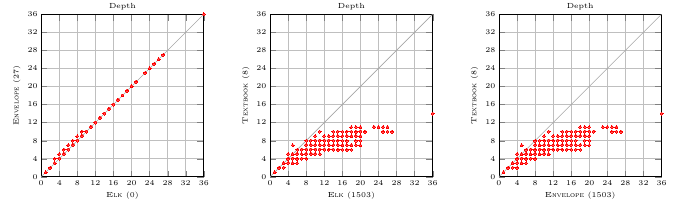}
	\caption{Depth}
	\label{fig:depth}
\end{figure}

\begin{figure}[ht]
	\centering
	\includegraphics[scale=1.3]{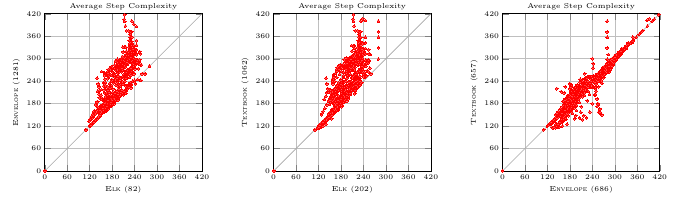}
	\caption{Average step complexity}
	\label{fig:avg}
\end{figure}

\begin{figure}[ht]
	\centering
	\includegraphics[scale=1.3]{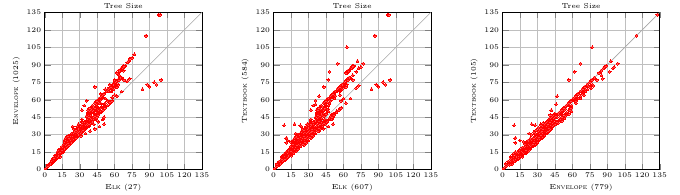}
	\caption{Tree Size}
	\label{fig:tree-size}
\end{figure}

\vfill
~

\end{document}